\documentclass[oribibl,envcountsame]{llncs}

\usepackage[normalem]{ulem}			%% Correction primitives
\usepackage{proof}					%% Drawing logical rules
\usepackage[nosumlimits]{amsmath}
\usepackage{amssymb}				%% Extra math symbols and macro
\usepackage{stmaryrd,latexsym}		%% More symbols
\usepackage{wasysym}				%% Yet more symbols
\usepackage{bm}						%% Bold math symbols\Gamma
\usepackage{ushort}					%% Single character underscores
\usepackage{slashed}				%% Сrossed symbols
\usepackage{colonequals}			%% For ':='
\usepackage{bm}						%% bold math

\usepackage{times}
\usepackage{microtype}
\usepackage{enumerate}

\usepackage{optparams}				%% Multiple optional paramenters
\usepackage{xspace}

\usepackage[pdfpagemode={UseOutlines},bookmarks=true,bookmarksopen=true,
   bookmarksopenlevel=0,bookmarksnumbered=true,hypertexnames=false,
   colorlinks,linkcolor={blue},citecolor={blue},urlcolor={red},
   pdfstartview={FitV},unicode,breaklinks=true,final]{hyperref}

\usepackage[all,cmtip]{xy}			%% Commutative diagrams

\usepackage{ifdraft}  				%% For access to the draft/final option of the documentclass
\ifdraft{
\usepackage[show]{ed}%ednotes
}{\usepackage[hide]{ed}}

\spnewtheorem{thm}[theorem]{Theorem}{\bfseries}{\itshape}
\spnewtheorem{cor}[theorem]{Corollary}{\bfseries}{\itshape}
\spnewtheorem{lem}[theorem]{Lemma}{\bfseries}{\itshape}
\spnewtheorem{prop}[theorem]{Proposition}{\bfseries}{\itshape}
\spnewtheorem{defn}[theorem]{Definition}{\bfseries}{\upshape}
\spnewtheorem{rem}[theorem]{Remark}{\bfseries}{\upshape}
\spnewtheorem{expl}[theorem]{Example}{\bfseries}{\upshape}
\spnewtheorem{thmdefn}[theorem]{Theorem and Definition}{\bfseries}{\itshape}
\spnewtheorem{propdefn}[theorem]{Proposition and Definition}{\bfseries}{\itshape}
\spnewtheorem{assumption}[theorem]{Assumption}{\bfseries}{\upshape}
\spnewtheorem{algorithm}[theorem]{Algorithm}{\bfseries}{\upshape}

\makeatletter
\newcommand{\SG}[2][]{\ed@note{#2}{SG}{#1}}
\newcommand{\LS}[2][]{\ed@note{#2}{LS}{#1}}
\makeatother

\usepackage{enumitem}
\setdescription{font=\it}
\setlist{topsep=1.2ex,itemsep=1.2ex}

\newcommand{\Cat}{\mathbf}

\newcommand{\BC}{{\Cat C}}

\newcommand{\BBT}{\mathbb{T}}

\newcommand{\Set}{\Cat{Set}}
\newcommand{\Cpo}{\Cat{Cpo}}

\newcommand{\Opname}{\mathrm}

\newcommand{\Ob}{{\Opname{Ob}\,}}

\newcommand{\id}{\operatorname{id}}
\newcommand{\PSet}{{\mathcal P}}

\newcommand{\True}{\top}
\newcommand{\False}{\bot}

\newcommand{\impl}{\Rightarrow}

\newcommand{\nil}{\varnothing}

\newcommand{\comp}{\mathbin{\circ}}
\newcommand{\gris}{\mathrel{::=}}

\newcommand{\argument}{\_\!\_}%{\underline{\;\;}}

\newcommand{\unit}{\star}
\newcommand{\fst}{\operatorname{\sf fst}}
\newcommand{\snd}{\operatorname{\sf snd}}

\newcommand{\inl}{\operatorname{\sf inl}}
\newcommand{\inr}{\operatorname{\sf inr}}
\newcommand{\inj}{\operatorname{\sf inj}}

\newcommand{\CASE}{\operatorname{\sf case}}
\newcommand{\OF}{\operatorname{\sf of}}

\ifx\case\undefined\def\case{}\fi
\renewcommand{\case}[3]{\CASE\kern1.2pt #1\kern1.2pt \OF\kern1.2pt #2;\kern1.2pt #3}

\newcommand{\caseOne}[2]{\CASE\kern1.2pt #1\kern1.2pt \OF\kern1.2pt #2}

\newcommand{\DO}{\operatorname{\sf do}}
\newcommand{\retOp}{\operatorname{\sf ret}}
\newcommand{\ret}[1]{\retOp #1}
\newcommand{\letTerm}[2]{\DO\kern1.2pt#1; #2}
\newcommand{\leteq}{\gets}
\newcommand{\letret}{\colonequals} %{\Lleftarrow}

\newcommand{\letTermO}[2]{\DO_{\nu}\kern1.2pt#1; #2}
\newcommand{\LOOP}{\operatorname{\sf unfold}}
\newcommand{\THEN}{\operatorname{\sf then}}

\makeatletter
\long\def\loopTerm@[#1][#2][#3]#4#5#6{
\DO\kern1.2pt#4 #1 \LOOP #3 #5 #2\THEN #6
}

\newcommand{\loopTerm}{
\optparams{\loopTerm@}{[;][\};][\{]}
}
\makeatother

\newcommand{\LET}{\operatorname{\sf init}}
\newcommand{\letLoop}[2]{\LET #1\kern1.2pt\LOOP\kern1.2pt\{#2\}}
\newcommand{\LetLoop}[2]{\LET #1\kern1.2pt\LOOP\kern1.2pt\bigl\{#2\bigr\}}

\newcommand{\fold}{\operatorname{\sf tuo}}
\newcommand{\unfold}{\operatorname{\sf out}}

\newcommand{\step}[1]{\hm{[} #1 \hm{]}}

\newcommand{\dispTerm}[2]{\operatorname{\sf disp}\kern1.2pt#1; #2}

\newcommand{\EE}{\scriptscriptstyle E}

\newcommand{\rcont}{\operatorname{\sf cont}}
\newcommand{\rdone}{\operatorname{\sf stop}}

\newcommand{\cont}{\operatorname{\sf rest}}
\newcommand{\done}{\operatorname{\sf done}}

\newcommand{\NEXT}{\operatorname{\sf next}}
\newcommand{\IS}{\operatorname{\sf is}}

\newcommand{\casenext}[3]{\NEXT\kern1.2pt #1\kern1.2pt \IS\kern1.2pt #2;\kern1.2pt #3}

\newcommand{\filter}{\operatorname{\sf filter}}

\newcommand{\await}{\operatorname{\sf await}}

\newcommand{\inv}{\operatorname{\sf inv}}
\newcommand{\exec}{\operatorname{\sf exec}}

\newcommand{\IF}{\operatorname{\sf if}}
\newcommand{\ifTerm}[3]{\IF #1\kern2.2pt {\sf then}\kern1.2pt #2\kern2.2pt {\sf else}\kern2.2pt #3}
\newcommand{\ifTermO}[3]{\IF_{\nu} #1\kern2.2pt {\sf then}\kern2.2pt #2\kern2.2pt {\sf else}\kern2.2pt #3}

\newcommand{\WHILE}{\operatorname{\sf while}}
\newcommand{\whileTerm}[3]{\LET #1\kern1.2pt\WHILE\kern1.2pt #2 \kern2.2pt{\sf do}\kern2.2pt #3}
\newcommand{\whileTermS}[2]{\WHILE\kern1.2pt #1 \kern2.2pt{\sf do}\kern2.2pt #2}

\newcommand{\UNTIL}{\operatorname{\sf until}}
\newcommand{\untilTerm}[3]{\LET #1\kern2.2pt{\sf do}\kern2.2pt #2\kern1.2pt\UNTIL #3}
\newcommand{\untilTermS}[2]{{\sf do}\kern2.2pt #1\kern1.2pt\UNTIL #2}

\newcommand{\SEQ}{\operatorname{\sf seq}}
\newcommand{\seqTerm}[2]{\SEQ\kern1.2pt#1; #2}

\newcommand{\lrule}[3]{\textbf{(#1)}\;\;\infrule{#2}{#3}}
\newcommand{\lrrule}[3]{\textbf{(#1)}\;\;\infrule{\underline{#2}}{#3}}

\newcommand{\infrule}[2]{\frac{#1}{#2}}

\newcommand{\lsem}{\llbracket}
\newcommand{\rsem}{\rrbracket}
\newcommand{\sem}[1]{\lsem #1 \rsem}

\newcommand{\brks}[1]{\langle #1\rangle}
\newcommand{\Brks}[1]{\bigl\langle #1\bigr\rangle}

\newcommand{\ana}[1]{\digamma_{\!\! #1}}
\newcommand{\fin}[1]{\alpha_{\scriptscriptstyle #1}}

\newlength{\myboxwidth}
\newenvironment{myfigure}{
\begin{figure*}[t]\begin{center}
\setlength{\fboxsep}{10pt}
}{
\end{center}\end{figure*}
}
	
\renewcommand{\ll}{\mathbin{\parallel}} 
\DeclareMathOperator{\lmOp}{\llfloor}
\newcommand{\lm}{\lmOp}
\DeclareMathOperator{\lsmOp}{\lfloor\kern0.6pt}
\newcommand{\lsm}{\lsmOp}
\DeclareMathOperator{\rsmOp}{\rfloor\kern0.6pt}
\newcommand{\rsm}{\rsmOp}

\newcommand*{\smallsub}[1]{
\mathchoice
 {\displaystyle{#1}}
 {\textstyle{#1}}
 {\scriptscriptstyle{#1}}
 {\scriptscriptstyle{#1}}
}

\newcommand{\nuM}{{\ensuremath{\smallsub{\mathsf{ME_{\boldsymbol\nu}}}}}\xspace}

\DeclareMathOperator{\redpar}{{
\declareslashed{}{
\vrule height2pt depth 2pt
\kern1pt
\vrule height2pt depth 2pt
}{0}{0}{\rightarrowtail}\slashed{\rightarrowtail}}}

\newcommand{\klstar}{\star}  					%% Kleisly star
\newcommand{\klcomp}{\mathbin{\diamond}}  		%% Kleisly composition

\newcommand{\Vars}{\mathrm{Vars}}

\newcommand{\setflag}{\mathrm{set\_flag}}
\newcommand{\getflag}{\mathrm{get\_flag}}

\newcommand{\setturn}{\mathrm{set\_turn}}
\newcommand{\turnis}{\mathrm{turn\_is}}

\newcommand{\incs}{\mathrm{in\_cs}}
\newcommand{\outcs}{\mathrm{out\_cs}}
\newcommand{\cs}{\mathrm{cs}}

\newcommand{\flip}{\mathrm{flip}}

\usepackage{ifpdf}					%% For switching between pdf and non-pdf mode
\ifx\todo\undefined
\def\todo{}
\fi
\ifpdf
  \usepackage[usenames,dvipsnames]{color} 
  \renewcommand{\todo}[1]{
  {\color{red}{[TODO: #1]}}
  }
\else
  \usepackage[usenames,dvips]{color}
  \renewcommand{\todo}[1]{
  {\color{red}{[TODO: #1]}}
  }
\fi

\newcounter{blubber}

\newenvironment{myitemize}
{\begin{itemize}
\setlength{\itemsep}{0cm}
\setlength{\parsep}{0cm}
}
{\end{itemize}}

\makeatletter
\def\subsection{%
\@startsection{subsection} % name of section and also name of
{2} % level, level 0 is chapter
{0 pt} % Indentation from left margin.
{2 pt plus 3 pt minus 4 pt}
{0 pt}
{\normalsize\bfseries}}
\makeatother

\title{A Coinductive Calculus for\\Asynchronous Side-effecting Processes}
\author{Sergey Goncharov\and Lutz Schr\"oder}
\institute{Safe and Secure Cognitive Systems, DFKI GmbH, Bremen}
\begin{document}
\maketitle
\begin{abstract}
  We present an abstract framework for concurrent processes in which
  atomic steps have generic side effects, handled according to the
  principle of monadic encapsulation of effects. Processes in this
  framework are potentially infinite resumptions, modelled using final
  coalgebras over the monadic base. As a calculus for such processes,
  we introduce a concurrent extension of Moggi's monadic metalanguage
  of effects. We establish soundness and completeness of a natural
  equational axiomatisation of this calculus. Moreover, we identify a
  corecursion scheme that is explicitly definable over the base
  language and provides flexible expressive means for the definition
  of new operators on processes, such as parallel composition. As a
  worked example, we prove the safety of a generic mutual exclusion
  scheme using a verification logic built on top of the equational
  calculus.
\end{abstract}

\section{Introduction}
%JonesHughesEtAl99
Imperative programming languages work with many different side
effects, such as I/O, state, exceptions, and others, which all
moreover come with strong degrees of variations in detail. This
variety is unified by the principle of monadic encapsulation of
side-effects~\cite{Moggi91}, which not only underlies extensive work
in semantics (e.g.~\cite{JacobsPoll03}) but,
following~\cite{Wadler97}, forms the basis for the treatment of
side-effects in functional languages such as
Haskell~\cite{PeytonJones03} and F\#~\cite{SymeGraniczEtAl07}. Monads
do offer support for concurrent programming, in particular through
variants of the \emph{resumption monad
  transformer}~\cite{CenciarelliMoggi93,Harrison06}, which lifts
resumptions in the style of Hennessy and
Plotkin~\cite{HennessyPlotkin79} to the monadic level, and which has
moreover been used in information flow security~\cite{HarrisonHook09},
semantics of functional logic programming languages such as
Curry~\cite{TolmachAntoy03}, modelling underspecification of
compilers, e.g.\ for ANSI C~\cite{PapaspyrouMacos00,Papaspyrou01}, and
to model the semantic of the
$\pi$-calculus~\cite{FioreMoggiEtAl02}. However, the formal basis for
concurrent functional-imperative programming is not as well-developed
as for the sequential case; in particular, Moggi's original
\emph{computational meta-language} is essentially limited to linear
sequential monadic programs, and does not offer native support for
concurrency.

The objective of the present work is to develop an extension of the
computational meta-language that can serve as a minimal common basis
for concurrent functional-imperative programming and semantics.  We
define an abstract \emph{meta-calculus for monadic processes} that is
based on the resumption monad transformer, and hence generic over the
base effects inherent in individual process steps. We work with
infinite resumptions, which brings tools from coalgebra into play, in
particular corecursion and coinduction~\cite{Rutten00}. We present a
complete equational axiomatization of our calculus which includes a
simple loop construct (in coalgebraic terms, coiteration) and then
derive a powerful \emph{corecursion schema} that allows defining
processes by systems of equations. It has a fully syntactic
justification, i.e.\ one can explicitly construct a solution to a
corecursive equation system by means of the basic term
language. Although there are strong corecursion results available in
the literature~(e.g.\ \cite{Bartels03,Uustalu03,MiliusPalmEtAl09}), our
corecursion schema does not seem to be covered by these, in particular
as it permits prefixing corecursive calls with monadic sequential
composition.

We exemplify our corecursion scheme with the definition of a number of
basic imperative programming and process-algebraic primitives
including parallel composition, and present a worked example, in which
we outline a safety proof for a monadic version of Dekker's mutual
exclusion algorithm, i.e.\ a concurrent algorithm with generic
side-effects. To this end, we employ a more high-level verification
logic that we develop on top of the basic equational calculus.

\paragraph {Further related work}
There is extensive work on axiomatic perspectives on effectful
iteration and recursion, including traced pre-monoidal
categories~\cite{BentonHyland03}, complete iterative
algebras~\cite{MiliusPalmEtAl09}, Kleene
monads~\cite{Goncharov10}, and recursive monadic binding
~\cite{ErkokLaunchbury00}.  The abstract notion of resumption goes
back at least to~\cite{HennessyPlotkin79}.  (Weakly) final coalgebras
of I/O-trees have been considered in the context of dependent type
theory for functional programming, without, however, following a fully
parametrised monadic approach as pursued
here~\cite{HancockSetzer05}. A framework where infinite resumptions of
a somewhat different type than considered here form the morphisms of a
category of processes, which is thus enriched over coalgebras for a
certain functor, is studied in~\cite{KrsticLaunchburyEtAl01}, but no
metalanguage is provided for such processes. A metalanguage that
essentially adds least fixed points, i.e.\ \emph{inductive} data types
as opposed to \emph{coinductive} process types as used in the present
work, to Moggi's base language is studied in~\cite{Filinski07};
reasoning principles in this framework are necessarily of a rather
different flavour. A resumption monad without a base effect, the
\emph{delay monad}, is studied in~\cite{Capretta05} with a view to
capturing general recursion. Our variant of the resumption monad
transformer belongs to the class of \emph{parametrized monads}
introduced in~\cite{Uustalu03}, where a form of corecursive scheme is
established which however does not seem to cover the one introduced
here.

\section{Computational Monads and Resumptions}
%The use of monads for the encapsulation of effects,
%following~\cite{Moggi91}, has become standard both in programming
%language semantics, e.g.\ for Java~\cite{JacobsPoll03}, C~\cite{Papaspyrou98} and in actual
%programming languages, as witnessed e.g.\ by the monad-based
%implementation of functional-imperative programming in
%Haskell~\cite{Peyton-Jones03}. 

We briefly recall the basic concepts of the monadic representation of
side-effects, and then present the specific semantic framework
required for the present work. Intuitively, a monad $\BBT$ (mostly
referred to just as $T$) associates to each type $A$ a type $TA$ of
computations with results in $A$; a function with side effects that
takes inputs of type $A$ and returns values of type $B$ is, then, just
a function of type $A\to TB$. In other words, by means of a monad we
can abstract from notions of computation by switching from non-pure
functions to pure ones with a converted type profile. One of the
equivalent ways to define a monad over a category $\BC$ is by giving a
\emph{Kleisli triple} $\BBT=(T,\eta,\argument^{\klstar})$ where
$T:\Ob\BC\to\Ob\BC$ is a function, $\eta$ is a family of morphisms
$\eta_A:A\to TA$ called \emph{unit}, and $\argument^\klstar$ assigns
to each morphism $f:A\to TB$ a morphism $f^\klstar:TA\to TB$ such that
\begin{math}\eta_A^\klstar     =       \id_{TA}, ~~~~
  f^\klstar\comp\eta_A              =       f, \textrm{ ~~~and}~~~~
  g^\klstar\comp f^\klstar          =       (g^\klstar\comp f)^\klstar.
\end{math}
%This leads to the \emph{Kleisli category} $\BC_\BBT$ of $\BBT$, which
%is a category having the same objects as $\BC$, and $\BC$-morphisms
%$A\to TB$ as morphisms $A\to B$, with \emph{Kleisli composition}
%$\klcomp$ defined by $f\klcomp g=f^\klstar\circ g$. 
Thus, $\eta_A$ converts values of type $A$ into side-effect free
computations, and $\argument^\klstar$ supports the sequential
composition $g^\klstar f$ of programs $f:A\to TB$ and $g:B\to
TC$. A monad over a Cartesian category is \emph{strong} if it is equipped with a natural transformation $\tau_{A,B}:A\times TB\to T(A\times B)$ called \emph{strength}, subject to certain coherence conditions~\cite{Moggi91}. %A strong monad hence can be given by a quadruple $(T,\eta,\mu,\tau)$. 
Since we are interested in concurrency, we require additional
structure for non-determinism:
\begin{defn}[Strong semi-additive monads]
  A strong monad $\BBT=(T,\eta,\argument^{\klstar},\tau)$ is
  \emph{semi-additive} if there exist natural transformations
  ${\delta:1\to T}$ and $\varpi:T\times T\to T$ making every $TA$ an internal
  bounded join-semilattice object so that $\argument^{\klstar}$ and $\tau$ respect the join-semilattice
structure in the following sense:
\begin{align*}
f^{\star}\delta    &=\delta,& f^{\star}\varpi &= \varpi\brks{f^{\klstar},f^{\klstar}},\\
\tau(f\times\delta)&=\delta,& \tau(f\times\varpi) &= \varpi\brks{\tau(f\times\pi_1),\tau(f\times\pi_2)}.
\end{align*}
\end{defn}
\noindent The above definition forces the nondeterministic choice to be an \emph{algebraic operation} in sense of~\cite{PlotkinPower01}. This implies that the
semilattice structure distributes over binding from the left (but not
necessarily from the right) as reflected in our calculus in Section~\ref{sec:calculus}.

% such that the following
%  diagram commutes:
%%
%\begin{displaymath}
%\xymatrix@R30pt@C30pt@M6pt{
%T^2A\times T^2A	\ar[r]^>>>>>>{\varpi_{TA}} \ar[d]_{\mu_A\times\mu_A} & T^2A\ar[d]^{\mu_A} & 1\ar[l]_>>>>>>{\delta_{TA}} \ar[dl]^>>>>>>>>>>>{\delta_{A}}\\
%TA\times TA \ar[r]^>>>>>>>{\varpi_A} & TA.
%}
%\end{displaymath}

\begin{rem}
As in the original computational metalanguage~\cite{Moggi91}, we work over an arbitrary base category $\BC$, with $\Set$ and $\omega\Cpo$  as prominent instances, thus establishing our corecursion scheme at a high level of generality. Our calculus will be sound and complete  w.r.t.\ the whole class of possible instances of $\BC$. Restricting, e.g., to order-theoretic models will, of course, preserve soundness, while completeness may break down due to particular properties of the base category becoming observable in the calculus. The FIX-logic of Crole and Pitts~\cite{CrolePitts90} is sound and complete w.r.t.\ certain order-theoretic models compatible with our models, in particular with $\Set$.
%The standard set-up using a variable base category thus amounts to a separation of concerns in that it allows for a study of coinductive principles in isolation. 
\end{rem}

\begin{expl}\hspace{-3pt}\cite{Moggi91}~~\label{expl:monads}
  The core examples of strong semi-additive monads are the finite
  powerset monad~$\PSet_\omega$, or, in the domain-theoretic setting,
  various powerdomain constructions.  Moreover, the powerset monad
  $\PSet$ and more generally, the \emph{quantale
    monad}~\cite{Jacobs10} $\lambda X.Q^X$ for a
  quantale~\cite{Rosenthal90} $Q$ are strong semi-additive
  monads. Further examples of semi-additive monads can be obtained
  from basic ones by combining them with other effects, e.g.\ by
  adding probabilistic choice~\cite{TixKeimelEtAl09} or by applying
  suitable monad transformers. In particular, the following monad
  transformers (which produce a new monad $Q$, given a monad $T$)
  preserve semi-additivity over any base category with sufficient structure:

  \vspace{-1em}
  \noindent
  \begin{tabular}{p{5cm}p{7cm}}
    \begin{enumerate}
    \item Exceptions: $Q A = T(A+E)$,
    \item States: $Q A = S\to T(A\times S)$,
    \end{enumerate}&
    \begin{enumerate}[start=3]
    \item I/O: $Q A = \mu X.\, T (A + I\to (X\times O))$,
    \item Continuations: $Q A = (A\to T K)\to T K$.
    \end{enumerate}
  \end{tabular}
  \vspace{-1em}
   
  \noindent E.g., the non-deterministic state monad $TX=S\to P(S\times
  X)$, is a strong semi-additive monad both over $\Set$ (with $P$
  denoting any variant of powerset) and over any reasonable category
  of domains (with $P$ denoting a powerdomain construction with deadlock).
\end{expl}
\noindent To model processes which are composed of atomic steps to be
thought of as pieces of imperative code with generic side-effects, we
use a variant of the \emph{resumption monad
  transformer}~\cite{CenciarelliMoggi93}: Assuming that for every
$X\in\Ob(\BC)$ the endofunctor $T(Id+X):\BC\to\BC$ possesses a final
coalgebra, which we denote
by $\nu\gamma.\,T(\gamma+X)$, we define a new monad $R$ by
\begin{equation*}
  RX=\nu\gamma.\,T(\gamma+X)
\end{equation*}
--- $R$ exists, e.g., if the base category is locally presentable and
$T$ is accessible~\cite{Worrell99}, a basic example being $TX=S\to
P(S\times X)$ where $P$ is finite powerset or a
powerdomain. Intuitively, a \emph{resumption}, i.e.\ a computation in
$RX$, takes an atomic step in $T$ and then returns either a value in
$X$ or a further computation in $RX$, possibly continuing in this way
indefinitely. Using a final coalgebra semantics amounts to identifying
processes up to coalgebraic behavioural equivalence, which generalizes
strong bisimilarity.

\section{A Calculus for Side-effecting Processes}\label{sec:calculus}
As originally observed by Moggi~\cite{Moggi91}, strong monads support a \emph{computational metalanguage}, i.e.\ essentially a generic sequential imperative programming language. Here we introduce a concurrent version of the metalanguage, the \emph{concurrent metalanguage}, based semantically on the resumption monad transformer. 

The concurrent metalanguage is parametrised over a countable signature $\Sigma$ including a set of atomic types $W$, from which the type system is generated by the grammar
\begin{displaymath}
P\gris W\mid 1 \mid P\times P\mid P+ P\mid T P\mid T_{\nu} P
\end{displaymath}
--- that is, we support sums and products, but not functional types,
our main target being the common imperative programming basis, which
does not include functional abstraction. Base effects are represented
by $T$, and resumptions by $T_\nu$. 

Moreover, $\Sigma$ includes function symbols $f:A\to B$ with given
profiles, where $A$ and $B$ are
types. %For purposes of this work, we require that the source type $A$ for $f$ is \emph{$T$-free}, i.e.\ does not mention $T$.
The terms of the language, also referred to as \emph{programs}, and their types are then determined by the rules shown in Fig.~\ref{fig:base-terms}; the dotted line separates operators for sequential non-determinism from the process operators. Besides the standard term language for sums and products and the bind and return operators $\DO$ and $\ret$ of the computational metalanguage, the concurrent metalanguage includes operations  $\nil$ and $+$ are called \emph{deadlock} and~\emph{choice}, respectively, as well as two specific constructs (\emph{out} and \emph{unfold}) for resumptions, explained later.
Judgements $\Gamma\rhd t:A$ read `term $t$ has type $A$ in context $\Gamma$',
where a \emph{context} is a list $\Gamma=(x_1:A_1,\dots,x_n:A_n)$ of
typed variables. Programs whose type is of the form $T_{\nu} A$ are called \emph{processes}.  %Here $\letTerm{x\leteq\argument}{\argument}$ is the only constructor binding variables. 
The notions of free and bound variables are defined in a standard way, as well as a notion of capture-avoiding substitution. %In the sequel we shall use the notation $(\letTerm{x\letret p}{q})$ as a shortening for $(\letTerm{x\leteq\ret p}{q})$.

\begin{myfigure} 
\fbox{\parbox{\myboxwidth}{\vspace{-3ex}
\begin{flalign*}
~~\lrule{var}
{x:A\in\Gamma}{\Gamma \rhd x:A}&&
\lrule{app}
{f:A\to B\in\Sigma ~~~~ \Gamma\rhd t:A}
{\Gamma\rhd f(t):B}
&&
\lrule{1}
{}
{\Gamma\rhd\unit:1}&&
\end{flalign*}\vspace{-2.2ex}
\begin{flalign*}
~~\lrule{pair}
{\Gamma\rhd t:A ~~~~
\Gamma\rhd u:B}
{\Gamma\rhd \langle t,u\rangle : A \times B}
&&
\lrule{fst}
{\Gamma\rhd t:A\times B}
{\Gamma\rhd\fst t:A}&&
\lrule{snd}
{\Gamma\rhd t:A\times B}
{\Gamma\rhd\snd t:B}&&
\end{flalign*}\vspace{-2.2ex}
\begin{flalign*}
~~\lrule{case}
{\Gamma\rhd s:A+B ~~~~ \Gamma,x:A\rhd t:C ~~~~ \Gamma,y:B\rhd u:C}
{\Gamma\rhd\case{s}{\inl x\mapsto t}{\inr y\mapsto u}:C}
&&
\lrule{nil}
{}
{\Gamma\rhd\nil:TA}&&
\end{flalign*}\vspace{-2.2ex}
\begin{flalign*}
~~
\lrule{inl}
{\Gamma\rhd t:A}
{\Gamma\rhd\inl t:A+B}
&&
\lrule{inr}
{\Gamma\rhd t:B}
{\Gamma\rhd\inr t:A+B}&&
\lrule{ret}
{\Gamma \rhd t:A}
{\Gamma \rhd \ret{t}:TA}&&
\end{flalign*}\vspace{-2.2ex}
\begin{flalign*}
~~\lrule{do}{\Gamma\rhd p:TA ~~~~
\Gamma, x:A\rhd q:TB}
{\Gamma\rhd \letTerm{x\leteq p}{q}:TB}&&
\lrule{plus}
{\Gamma\rhd p+ q:TA}
{\Gamma\rhd p:TA\quad\Gamma\rhd q:TA}&&
\end{flalign*}\vspace{-1ex}
\dotfill
\begin{flalign*}
~~\lrule{out}
{\Gamma\rhd p: T_{\nu} A}
{\Gamma\rhd\unfold(p): T(T_{\nu} A+A)}&&
\lrule{unf}
{\Gamma\rhd p: A ~~~\Gamma, x:A\rhd q: T(A+B)}
{\Gamma\rhd\letLoop{x\letret p}{q}:T_{\nu} B}&&
\end{flalign*}\vspace{-3ex}
}}
 \caption{Typing rules for the concurrent metalanguage}
 \label{fig:base-terms}
\vspace{-2em}
\end{myfigure}

The semantics of the concurrent metalanguage is defined over
\emph{$\nuM$-models}, referred to just as \emph{models} below. A model
is based on a distributive category~\cite{Cockett93} $\BC$, i.e.\ a
category with binary sums and finite products such that the canonical
map $A\times B+A\times C\to A\times(B+C)$ is an isomorphism, with
inverse $\mathit{dist}:A\times(B+C)\to A\times B+A\times C$ (this
holds, e.g., when $\BC$ is Cartesian closed). Moreover, it
specifies a strong semi-additive monad $T$ on $\BC$ such that for
every $A\in\Ob(\BC)$ the functor $T(Id+A)$ possesses a final coalgebra
denoted $RA=\nu\gamma.\,T(\gamma+A)$, thus defining a functor $R$
(\emph{resumptions}).

A model interprets base types as objects of $\BC$. The interpretation
$\sem{A}$ of types $A$ is then defined by standard clauses for $1$,
$A\times B$, and $A+B$ and $\sem{TA}=T\sem{A}$, $\sem{T_\nu
  A}=R\sem{A}$. For $\Gamma=(x_1:A_1,\dots,x_n:A_n)$ we put
$\sem{\Gamma}=\sem{A_1}\times\dots\times\sem{A_n}$.
Moreover, a  model interprets function symbols $f:A\to B$ as morphisms $\sem{f}:\sem{A}\to\sem{B}$, which induces an interpretation $\sem{t}:\sem{\Gamma}\to\sem{A}$ of programs $\Gamma\rhd
t:A$ given by the usual clauses for variables, function application, pairing, projections, injections, and $\unit$. The operations $+$ and $\nil$ are interpreted by the bounded join semilattice operations $\varpi$ and $\delta$ of $T$, respectively. For the monad operations and the case operator, we have
\begin{myitemize}
%\item $\sem{x_1:A_1,\ldots,x_n:A_n \rhd x_i:A_i} = \pi^n_i$,
%\item $\sem{\Gamma\rhd f(t):B} = \sem{f}\comp\sem{\Gamma\rhd t:A}$, $f:A\to B\in\Sigma$, 
%\item $\sem{\Gamma\rhd \unit:1} =\operatorname{!}_{\sem{\Gamma}}$,
%\item $\sem{\Gamma\rhd \brks{t,u}: A \times B}=\brks{\sem{\Gamma\rhd t:A},\sem{\Gamma\rhd u:B}}$,
%\item $\sem{\Gamma\rhd \fst t:A} = \pi_1\comp\sem{\Gamma\rhd t:A\times B}$, $\sem{\Gamma\rhd \snd t:A} = \pi_2\comp\sem{\Gamma\rhd t:A\times B}$,
%\item $\sem{\Gamma\rhd \inl t:A+B} = \kappa_1\comp\sem{\Gamma\rhd t:A}$, $\sem{\Gamma\rhd \inr t:A+B} = \kappa_2\comp\sem{\Gamma\rhd t:B}$,
\item $\sem{\Gamma\rhd\case{s}{\inl x\mapsto t}{\inr y\mapsto u}:C}=$\\$\bigl[\sem{\Gamma,x:A\rhd t:C}, \sem{\Gamma,y:B\rhd u:C}\bigr]\comp \mathit{dist}\comp\Brks{\id,\sem{\Gamma\rhd s:A+B}}$,
\item $\sem{\Gamma\rhd \letTerm{x\leteq p}{q}:TB} =\sem{\Gamma, x:A\rhd q:TB}\klcomp\tau_{\sem{\Gamma},\sem{A}}\comp\brks{\id, \sem{\Gamma\rhd p:TA}}$,
\item $\sem{\Gamma\rhd \ret t:TA}=\eta_A\circ\sem{\Gamma\rhd t:A}$,
%\item $\sem{\Gamma\rhd p+q: T A} = \varpi_{\sem{A}}\comp\brks{\sem{\Gamma\rhd p: T A}, \sem{\Gamma\rhd q: T A}}$,
%\item $\sem{\Gamma\rhd\nil: T A} = \delta_{\sem{A}}\comp !_{\sem{\Gamma}}$
\end{myitemize}
where as usual %the $\pi_i^n$ are the projections
%$A_1\times\cdots\times A_n\to A_i$, $\operatorname{!}_A:A\to 1$ is the
%unique morphism into the terminal object, 
$\langle f,g\rangle:A\to B\times C$ denotes pairing of morphisms
$f:A\to B$, $g:A\to C$, %the $\kappa_i$ with $i=1,2$ denote coproduct
%injections $A\to A+B$ and $B\to A+B$, 
and $[f,g]:A+B\to C$ denotes copairing of $f:A\to C$ and $g:B\to C$.

It remains to interpret $\unfold$, which is just the final coalgebra structure of $R A$, and the loop construct $\letLoop{x\letret p}{q}$ which captures coiteration. Formally, let $\fin{A}:RA\to T(RA+A)$  be the final coalgebra structure, and for a coalgebra $f:X\to T(X+A)$, let $\ana{f}:X\to R A$ be the unique coalgebra morphism. Then we put
\begin{align*}
	\sem{\Gamma\rhd\unfold(p):T(T_{\nu} A+A)}&=\fin{\sem{A}}\comp\sem{\Gamma\rhd p:T_{\nu} A}\\
	\sem{\Gamma\rhd\letLoop{x\letret p}{q}:T_{\nu} B} & =
          R\pi_2\comp\ana{f}\comp{\brks{\id,\sem{\Gamma\rhd p:A}}}
        \end{align*}
        where $f=T(\mathit{dist})\comp\tau\brks{\pi_1,g}$ with
        $g=\sem{\Gamma,x:A\rhd q:T(A+B)}$. Thus, $\ana{f}$ is uniquely
        determined by the commutative diagram\ednote{LS: this may have
          to go.}
\begin{displaymath}
\xymatrix@R30pt@C35pt@M6pt{
\sem{\Gamma}\times\sem{A}	\ar[r]^>>>>>>>>>{T(\mathit{dist})\comp\tau\brks{\pi_1,g}} \ar[d]_{\ana{f}} & T(\sem{\Gamma}\times\sem{A}+\sem{\Gamma}\times\sem{B})\ar[d]^{T(\ana{f}+\id)}\\
R(\sem{\Gamma}\times\sem{B}) \ar[r]^>>>>>>>{\fin{\sem{\Gamma}\times\sem{B}}} & T(R(\sem{\Gamma}\times\sem{B})+\sem{\Gamma}\times\sem{B}).
}
\end{displaymath}
\noindent 
A model is said to \emph{satisfy} a well-typed equation
$\Gamma\rhd t=s$ if $\sem{\Gamma\rhd t:A}=\sem{\Gamma\rhd s:A}$.

\begin{myfigure}
\fbox{\parbox{\myboxwidth}{\vspace{-3ex}
\begin{flalign*}
~~\textbf{(case\_inl)}\quad\case{\inl p}{\inl x\mapsto q}{\inr y\mapsto r} = q[p/x]&&\textbf{(fst)}\quad\fst\brks{p,q} = p~
\end{flalign*}\vspace{-4.6ex}
\begin{flalign*}
~~\textbf{(case\_inr)}\quad\case{\inr p}{\inl x\mapsto q}{\inr y\mapsto r} = r[p/y]&&\textbf{(snd)}\quad\snd\brks{p,q} =q~
\end{flalign*}\vspace{-4.6ex}
\begin{flalign*}
~~\textbf{(case\_id)}\quad\case{p}{\inl x\mapsto\inl x}{\inr y\mapsto\inr y} =p&&\textbf{(pair)}\quad\brks{\fst p,\snd p} = p~
\end{flalign*}\vspace{-3.6ex}
\begin{flalign*}
~~\textbf{(case\_sub)}\quad&\begin{split}\case{p}{\inl x\mapsto t[q/z]}{\inr y\mapsto t[r/z]}\\[1ex]
\quad=t[\case{p}{\inl x\mapsto q}{\inr y\mapsto r}/z]\end{split}\qquad&&&(x,y\notin\Vars(r))~
\end{flalign*}\vspace{-3ex}
\begin{flalign*}
~~\textbf{($\star$)}~~ p:1 =\unit\!&&\textbf{(unit$\bf_1$)}~~\letTerm{x\leteq p}{\ret{x}} = p&&
\textbf{(unit$\bf_2$)}~~\letTerm{x\leteq\ret{a}}{p} = p[a/x]
\end{flalign*}\vspace{-4.6ex}
\begin{flalign*}
~~\textbf{(assoc)}\quad\letTerm{x\leteq (\letTerm{y\leteq p}{q})}{r} &= \letTerm{x\leteq p;y\leteq q}{r}&&(y\notin\Vars(r))~
\end{flalign*}\vspace{-5ex}\par\dotfill\vspace{-1ex}
\begin{flalign*}
\quad\textbf{(nil)} 		\quad  p+ \nil & = p&
\textbf{(comm)} 		\quad p+ q & = q + p&
\textbf{(idem)}	    	\quad p+ p 	& = p&~
\end{flalign*}\vspace{-4.5ex}
\begin{flalign*}
\quad\textbf{(assoc\_plus)}  \quad p+(q+ r) & = (p+ q)+ r&
&&\textbf{(dist\_nil)} \quad \letTerm{x\leteq\nil}{r} &= \nil&~
\end{flalign*}\vspace{-4.5ex}
\begin{flalign*}
\quad\textbf{(dist\_plus)}   \quad \letTerm{x\leteq (p+ q)}{r}& = \letTerm{x\leteq p}{r}+\letTerm{x\leteq q}{r}&
\end{flalign*}\vspace{-5ex}\par\dotfill
\vspace{-1ex}
\begin{align*}
\quad\lrrule{co-iter}
{~\unfold(p) = \casenext{q}{\cont x\mapsto\rcont p}{\done y\mapsto\rdone y}~}
{p[y/x] = \letLoop{x\letret y}{q}}&&~
\end{align*}
\vspace{-3ex}
}}
  \caption{Axiomatization of the concurrent metalanguage}
  \label{fig:proc_calc}
\vspace{-2em}
\end{myfigure}
As suggestive abbreviations for use in process definitions, we write $\rcont$ for $(\ret\inl)$ and  $\rdone$ for $(\ret\inr)$. Moreover, we write $(\casenext{p}{\cont x\mapsto q}{\done y\mapsto r})$ for 
%\begin{displaymath}
%
$
(\letTerm{z\leteq p}{\case{z}{\inl x\mapsto q}{\inr y\mapsto r}})
$
.
%%\end{displaymath}
%
%
%
We also define a converse $\fold:T(T_\nu A+A)\to T_\nu A$ to $\unfold$ by 
\begin{equation*}
\fold(p)~=~\letLoop{q\letret p}{\casenext{q}{\cont y\mapsto\rcont(\unfold(y))}{\done x\mapsto\rdone x}}.
\end{equation*}
%
%Here and in the sequel, $\rdone:A\to T(T_{\nu}A+A)$ is a shorthand for $(\ret\inr)$ and $\rcont:T_{\nu}A\to T(T_{\nu} A+A)$ is a shorthand for $(\ret\inl)$.
%
% 
% Every program ${p:TA}$ now can be turned into a one-step process $\step(p):T_{\nu}A$ defined by the assignment: $\step(p):=\fold(\letTerm{x\leteq p}{\ret\inr x})$.
%
An axiomatization \nuM of the concurrent metalanguage is
given in Fig.~\ref{fig:proc_calc} (where we omit the standard
equational logic ingredients including the obvious congruence
rules). Apart from the standard axioms for products and coproducts,
\nuM contains three well-known monad laws, axioms for semi-additivity
(middle section) and a novel (bidirectional) rule~\textbf{(co-iter)} for
effectful co-iteration (which in particular can be used to show that
$\fold$ is really inverse to $\unfold$).
%\suppressfloats[t]
%
% There will also Hence, we call
%this equivalence~\emph{local}\SG{Mention bisimilarity?}. More advanced
%equivalence relations can be introduced on top of this coarse
%equivalence. A natural way to introduce such an equivalence is by
%actual running of processes. We expect the classical theory of to be
%well-adaptable to this goal.
%
%
%
\begin{theorem}\label{thm:rec_compl}
$\nuM$ is sound and strongly complete over $\nuM$-models.
\end{theorem}
%
%
%\begin{theorem}[Soundness and completeness of~$\nuM$]\label{thm:rec_compl}
%  The calculus~\nuM is sound and strongly complete.
%\end{theorem}
%
%By Theorem~\ref{thm:rec_compl}, ${\nu\gamma.\,T(\gamma+ Id)}$ is also a semi-additive monad \ednote{SG:turns out to be not so easy}. Note that the property of semi-additivity of $T$ is crucial here, e.g.\ if $T$ is the identity monad then ${\nu\gamma.\,T(\gamma+ Id)}$ is not even a monad.
%
%
A core result on the concurrent metalanguage is an
expressive corecursion scheme supported by the given simple
axiomatisation. Its formulation requires $n$-ary coproducts
$A_1+\dots+A_n$ with coproduct injections $\inj_i^n:A_i\to
{A_1+\ldots+A_n}$ and a correspondingly generalized case construct;
all this can clearly be encoded in $\nuM$.
\begin{thmdefn}[Mutual corecursion]\label{lem:solve}
  Let $f_i:A_i\to T_{\nu} B_i$, $i=1,\dots,k$ be fresh function
  symbols. A \emph{guarded corecursive scheme} is a system of equations
\begin{equation*}\label{eq:corec}
\unfold(f_i(x)) = \letTerm{z\leteq p_i}{\case{z}{\inj_1^{n_i} x_1\mapsto p_{1}^i}{\dotsc;\inj_{n_i}^{n_i} x_{n_i}\mapsto p_{n_i}^i}}
\end{equation*}
for $i=1,\dots,k$ such that for every $i$, $p_i$ does not contain any
$f_j$, and for every $i,j$ $p_j^i$ either does not contain any $f_m$
or is of the form $p_j^i\equiv\rcont f_m(x_j)$ for some $m$. Such a
guarded corecursive scheme uniquely defines $f_1,\ldots,f_k$ (as
morphisms in the model), and the solutions $f_i$ are expressible as
programs in $\nuM$.
\end{thmdefn}
As a first application of guarded corecursive schemes, we define a
binding operation $\DO_\nu$ with the same typing as $\DO$ but with $T$
replaced by $T_\nu$ corecursively by
\begin{equation*}
\unfold(\letTermO{x\leteq p}{q}) =\casenext{\unfold(p)}{\cont x\,\mapsto\rcont (\letTermO{x\leteq p}{q})}{\done x\mapsto\unfold(q)}.
\end{equation*}
Similarly, we define operations $\ret_{\nu}$, $\nil_{\nu}$, and
$+_\nu$ as analogues of $\ret$, $\nil$, and $+$ by putting $\ret_{\nu}
p=\fold(\rdone p)$, $\nil_{\nu}=\fold(\nil)$, and $p+_{\nu} q
=\fold(\unfold(p)+\unfold(q))$.  \emph{These operations turn $T_\nu$
  into a strong-semiadditive monad}; formally, we can derive (in
$\nuM$) the top and middle sections of Fig.~\ref{fig:proc_calc} with
$T$ replaced by $T_\nu$ (the monad laws already follow from results
of~\cite{Uustalu03}).

%It
%turns out that via these definitions, a semi-additive structure on $T$
%indeed induces a semi-additive structure on $T_{\nu}$. The following
%result is analogous to Lemma~\ref{lem:derived_monad} and can be
%verified by routine calculations.

% \begin{rem}
%   Resumptions may in a first approximation be thought of as finite or
%   infinite streams of atomic computation steps (when non-determinism
%   enters the picture, it will be more appropriate to think of them as
%   finite or infinite trees), and in fact bit streams may be regarded
%   as a very simple example of resumptions. It has been shown that when
%   arbitrary equational specifications of streams are admitted, then
%   equality of streams is $\Pi^0_2$-complete, in particular neither
%   r.e.\ nor co-r.e.~\cite{Rosu06}. Consequently, any reasonably
%   expressive metalanguage for monadic resumptions may be expected to
%   be computationally hard. The above completeness result implies that
%   our metalanguage for resumptions is r.e. Although this result is
%   partly based on admitting arbitrary base categories, it gives an
%   indication that what is seen somewhat pessimistically
%   in~\cite{Rosu06}, namely to identify a definition language for
%   streams and more generally resumptions with better computational
%   properties than unrestricted equational definitions, may actually be
%   feasible after all.
% \end{rem}

\section{Programming with Side-effecting Processes}

Above, we have begun to define operations on processes; in particular,
$\nil_{\nu}$ is a deadlocked process, and $+_{\nu}$ is a
nondeterministic choice of two processes. We next show how to define
more complex operations, including parallel composition, by means
of guarded corecursive schemes.

Note that over distributive categories, one can define the type of
Booleans with the usual structure as $2=1+1$. We write
$(\ifTerm{b}{p}{q})$ as an abbreviation for ${(\case{b}{\inl x\mapsto
    p}{\inr x\mapsto q})}$ where $b$ has type $2$.
\vspace{-3ex}

\subsubsection{Sequential composition} Although $T_\nu$ is a monad,
its binding operator is not quite what one would want as sequential
composition of processes, as it merges the last step of the first
process with the first step of the second process.  We can, however,
capture sequential composition (with the same typing) in the intended
way by putting
\begin{displaymath}
{\seqTerm{x\leteq p}{q}=\letTermO{x\leteq p}{\fold(\rcont q)}}.
\end{displaymath} 

\vspace{-3ex}
\subsubsection{Branching}
Using the effect-free $\operatorname{\sf if}$ operator defined
earlier, we can define a conditional branching operator for processes
$p,q:T_\nu A$ and a condition $b:T2$ by
$$\ifTermO{b}{p}{q}
=\fold\bigl(\letTerm{z\leteq b}{\ifTerm{z}{(\rcont p)}{(\rcont q)}}\bigr).$$

\vspace{-4ex}
\subsubsection{Looping} For terms $\Gamma\rhd p$; $\Gamma,x: A\rhd b: T2$; and
$\Gamma,x:A\rhd q: T_{\nu} A$, we define loops
$$\Gamma\rhd\whileTerm{x\letret p}{b}{q}: T_{\nu} A\quad\text{ and }\quad
\Gamma\rhd\untilTerm{x\letret p}{q}{b}: T_{\nu} A$$ as follows. We
generalize the until loop to a program $U^b_{x,q}(r)$ for
$\Gamma\rhd r:T_\nu A$ intended to represent $\seqTerm{y\leteq
  r}{\whileTerm{x:=y}{b}{q}}$ (so that $(\untilTerm{x\letret
  p}{q}{b})=U^b_{x,q}(q[p/x])$) and abbreviate
$W^b_{x,q}(p)=(\whileTerm{x\letret p}{b}{q})$. We then define the four
functions $W^b_{x,q}$ and $U^b_{x,q}$ (for $b:2$) by the guarded
corecursive scheme
\begin{flalign*}\renewcommand{\arraystretch}{1}
\unfold(W^b_{x,q}(p))=&\,\letTerm{v\leteq b[p/x]}{\ifTerm{v}{\rcont(U_{x,q}^{\neg b}(q[p/x]))}{\rdone(p)}},\\
%\unfold(U_{x,q}^b(r))=&\,\begin{array}[t]{rl}\casenext{\unfold(r)}{\cont z&\mapsto\rcont(U_{x,q}^{b}(z))}{\\\done y&\mapsto\rcont(W_{x,q}^{\neg b}(y))}.\end{array}
\unfold(U_{x,q}^b(r))=&\,\casenext{\unfold(r)}{\cont z\mapsto\rcont(U_{x,q}^{b}(z))}{\done y\mapsto\rcont(W_{x,q}^{\neg b}(y))}.
\end{flalign*}

\vspace{-2ex}
\subsubsection{Exceptions} As the concurrent metalanguage includes
coproducts, the exception monad transformer($T^{\EE}
A=T(A+E)$)~\cite{CenciarelliMoggi93} and the corresponding operations for raising and handling exceptions are directly expressible in $\nuM$.

\vspace{-2ex}
\subsubsection{Interleaving}
We introduce process interleaving $\ll:T_{\nu}A\times T_{\nu}B\to
T_{\nu}(A\times B)$ by a CCS-style expansion law~\cite{Milner89} (using
an auxiliary left merge $\lm$)
\begin{flalign*}
\unfold(p\ll q) =\;&\unfold(p\lm q)+ \letTerm{\brks{x,y}\leteq\unfold(q\lm p)}{\ret\brks{y,x}},\\
\unfold(p\lm q) =\;&\begin{array}[t]{rl}\casenext{\unfold(p)}{\cont r&\mapsto\rcont(r\ll q)}{\\\done x&\mapsto\rcont(\letTermO{y\leteq q}{\ret_\nu\brks{x,y}})}.\end{array}
\end{flalign*}
This is easily seen to be equivalent to the guarded corecursive
scheme
\begin{flalign*}
\unfold(p\ll q) =\;&\letTerm{u\leteq (p\lsm q + p\rsm q)}{\\&\case{u}{\inl\brks{s,t}\mapsto\rcont(s\ll t)}{\inr r\mapsto\rcont r}}
\end{flalign*}
where for $p:T_{\nu} A$, $q:T_{\nu} B$, $p\lsm q:T(T_{\nu} A\times
T_{\nu} B + T_{\nu}(A\times B))$ is defined as
\begin{flalign*}
p\lsm q =\;&\casenext{\unfold(p)}{\cont r\mapsto\ret\inl\brks{r,q}}{\done x\mapsto\ret\inr(\letTermO{y\leteq q}{\ret_{\nu}\brks{x,y}})}
\end{flalign*}
and $p\rsm q:T(T_{\nu} A\times T_{\nu} B + T_{\nu}(A\times B))$ is the evident dual of $p\lsm q$.
\section{Verification and Process Invariants}
We now explore the potential of our formalism as a \emph{verification}
framework, extending existing monad-based program
logics~\cite{SchroderMossakowski04b,SchroderMossakowski09} to
concurrent processes. A cornerstone of these frameworks is a notion
of \emph{pure} program:
\begin{defn}[Pure programs]
  A program $p:T A$ is \emph{pure} if 
\begin{myitemize}
\item $p$ is \emph{discardable}, i.e., $\letTerm{y\leteq
    p}{\ret{\unit}}=\ret{\unit}$;
\item $p$ is \emph{copyable}, i.e.\ $ \letTerm{x\leteq p;y\leteq
    p}{\ret\brks{x,y}}=\letTerm{x\leteq p}{\ret\brks{x,x}}$; and
\item $p$ commutes with any other discardable and copyable program $q$, i.e.\\
\begin{math}
    (\letTerm{x\leteq p;y\leteq q}{\ret\brks{x,y}})=
    \letTerm{y\leteq q;x\leteq p}{\ret\brks{x,y}}.
\end{math}
\end{myitemize}
\end{defn}
Intuitively, pure programs are those that can access internal data
behind the computation but cannot affect it. A typical example of a
pure program is a getter method. As shown
in~\cite{SchroderMossakowski04b}, pure programs form a submonad $P$ of
$T$. A~\emph{test} is a program of type $P2$. All logical
connectives extend to tests; e.g.\ $\neg b = (\letTerm{x\leteq
  b}{\ret\neg x})$ for $b:P2$.

Given a program $p:T A$ and tests $\phi,\psi:P2$, the program
$\filter(p,\phi,\psi):T A$ is defined by the equation
\begin{align*}
\filter(p,\phi,\psi)=\letTerm{x\leteq\phi}{y\leteq p;z\leteq\psi;\ifTerm{(x\impl z)}{\ret y}{\nil}}.
\end{align*}
Intuitively, $\filter$ modifies the given program $p$ by removing
those threads that satisfy the precondition $\phi$ but fail the
postcondition $\psi$. This enables us to encode a Hoare triple
(alternatively to~\cite{SchroderMossakowski04b,SchroderMossakowski09})
by the equivalence
\begin{equation*}
\{\phi\}p\{\psi\}\iff\filter(p,\phi,\psi) = p
\end{equation*}
--- i.e.\ the Hoare triple $\{\phi\}p\{\psi\}$ is satisfied iff
$\filter(p,\phi,\psi)$ does not remove any execution paths from
$p$. On the other hand, $\filter$ extends to processes as follows:
\begin{align*}
\filter_{\nu}(p,\phi,\psi)=\letLoop{z\letret p}{\fold(\filter(\unfold(z),\phi,\psi))}.
\end{align*}
It turns out that the above definition of Hoare triple is equivalent
to the one from~\cite{SchroderMossakowski04b,SchroderMossakowski09},
which in particular enables use of the sequential monad-based Hoare
calculus of~\cite{SchroderMossakowski09}:
\begin{lemma}\label{lem:hoare}
  For every program $p$ and tests $\phi$, $\psi$, $\{\phi\}p\{\psi\}$
  is equivalent to the equation
\begin{align*}
&\letTerm{x\leteq\phi;y\leteq p;z\leteq\psi}{\ret\brks{x,y,z,x\impl z}}=\\
&\letTerm{x\leteq\phi;y\leteq p;z\leteq\psi}{\ret\brks{x,y,z,\True}}.
\end{align*}
\end{lemma}

\noindent A test $\phi$ is an \emph{invariant} of process $p$ if
$\filter_{\nu}(p,\phi,\phi)=p$. We use $\inv(p,\phi)$ as a shorthand
for this equality. Given a process $p:T_{\nu} A$, we define
\emph{partial execution} of $p$ by
\begin{align*}
\exec(p) = \fold(\casenext{p}{\cont x\mapsto\unfold(x)}{\done x\mapsto\rdone x}).
\end{align*}
For every $p$, $\exec(p)$ is precisely the program obtained by
collapsing the first and the second steps of $p$ into one. We denote
by $\exec^n(p)$ the $n$-fold application of $\exec$ to $p$. This
allows us to formalize satisfaction of a safety property $\phi$ by a
process~$p$:
\begin{center}
\it `$p$ is safe w.r.t.\ $\psi$ at $\phi$' iff for every $n$, $\{\phi\}\exec^n(p)\{\psi\}$.
\end{center}
Note, however, that this definition is not directly expressible in our
logic, because it involves quantification over the naturals. Often
this problem can be overcome by picking out a suitable process
invariant.

\begin{lem}\label{lem:safety}
  Let $\phi,\psi$, and $\xi$ be tests such that $\phi\impl\xi$
  and $\xi\impl\psi$, and let $p$ be a process. Then
  $\inv(p,\xi)$ implies $\{\phi\}\exec^n(p)\{\psi\}$ for every $n$.
\end{lem}

\section{Worked Example: Dekker's Mutual Exclusion Algorithm}
We illustrate the use of our calculus by encoding Dekker's mutual
exclusion algorithm. This algorithm was originally presented as an
Algol program, and hence presumes some fixed imperative semantics,
while we present (and verify) a version with \emph{generic}
side-effects. We introduce the following signature symbols:
\begin{flalign*}
&&\setflag:&~2\times2\to T1,	&\setturn:&~2\to T1,&&\\
&&\getflag:&~2\to P2,			&\turnis:&~2\to P2,&&
\end{flalign*}
which can be roughly understood as interface functions accessing
variables \texttt{flag1},~\texttt{flag2} and~\texttt{turn}. This is
justified by a suitable equational axiomatization of the above
operators, which includes the following axioms (we assume $i\neq j$):
\begin{flalign*}
\letTerm{\setflag(i,b)}{\getflag(i)}   &= \letTerm{\setflag(i,b)}{\ret b}\\
\letTerm{\setflag(i,b)}{\setflag(j,c)} &= \letTerm{\setflag(j,c)}{\setflag(i,b)}\\
\letTerm{\setflag(i,b)}{\setflag(i,c)} &= \setflag(i,c)\\[1ex]
\letTerm{\setturn(b)}{\turnis(c)}      &= \letTerm{\setturn(b)}{\ret(b\Leftrightarrow c)}\\
\letTerm{\setturn(b)}{\setturn(c)} &= \setturn(c).
\end{flalign*}
(Obvious further axioms are omitted.)  The crucial part of Dekker's
algorithm is a (sub)program implementing \emph{busy waiting}. In our
case this is captured by the function $busy\_wait: 2\to T 1$ defined
as follows:
\begin{flalign*}
busy\_wait(i) = \whileTermS{\,&\getflag(\mathit{\flip}(i))}{\ifTermO{\turnis(\mathit{\flip}(i))\\&}{
\seqTerm{\step{\setflag(i,\False)}}{\await(\turnis(i))}\\&}{\step{\setflag(i,\True)}}}
\end{flalign*}
Here, we used the following shorthands: $\step{p}=\fold(\rdone p)$
denotes the one-step process defined by $p:T A$, $\mathit{\flip}:2\to 2$ is the
function swapping the coproduct components of $2=1+1$;
$(\whileTermS{b}{q})$ encodes $(\whileTerm{x\leteq\unit}{b}{q})$;
finally, $(\await{b})$ with $b$ of type $T2$, intuitively meaning
`wait until $b$', is defined by the equation:
\begin{displaymath}
\await{b} = \whileTermS{\neg b}{\ret_{\nu}\unit}.
\end{displaymath} 
Finally, we define a generic process accessing the critical section:
\begin{flalign*}
proc(i,p) = \seqTerm{&\step{\setflag(i,\True)}}{busy\_wait(i);\\ &\step{\incs(i)}; p; \step{\outcs(i)};\\ &\step{\setturn(\mathit{\flip}(i))}; \step{\setflag(i,\False)}}.
\end{flalign*}
Here we use the functions $\incs,\outcs:2\to T1$ in order to keep track of the beginning and the end of the critical section. These functions are supposed to work together with the testing function $\cs:2\to T2$ as prescribed by the axioms
\begin{flalign*}
\letTerm{\incs(i)}{\cs(i)}=&~\letTerm{\incs(i)}{\ret\True},\\
\letTerm{\outcs(i)}{\cs(i)}=&~\letTerm{\outcs(i)}{\ret\False}.
\end{flalign*}
Now the safety condition for the algorithm can be expressed by the formula
\begin{equation*}
\forall n.\,\{\neg\cs(\bar 1)\land\neg\cs(\bar 2)\}\exec^n(proc(\bar 1,p)\ll proc(\bar 2,q))\{\neg\cs(\bar 1)\lor\neg\cs(\bar 2)\}
\end{equation*}
where $\bar 1$ and $\bar 2$ are the canonical coproduct injections $\inl\unit$ and $\inr\unit$. By Lemma~\ref{lem:safety}, it suffices to show that the following formula
\begin{gather*}\label{eq:inv}
\begin{split}
&\neg\cs(\bar 1)\land\cs(\bar 2)\land\getflag(\bar 2)~\lor
\neg\cs(\bar 2)\land\cs(\bar 1)\land\getflag(\bar 1)~\lor
\neg\cs(\bar 1)\land\neg\cs(\bar 2)\phantom{~\lor}
\end{split}
\end{gather*}
is an invariant of $proc(\bar 1,p)\ll proc(\bar 2,q)$. As can be shown by definition of parallel composition, this holds iff the same formula is an invariant of both $proc(\bar 1,p)$ and $proc(\bar 2,q)$, which in turn can be shown by coinduction in~\nuM.

\section{Conclusions and further work}\LS{I need to rewrite this.}
We have studied asynchronous concurrency in a framework of generic
effects. To this end, we have combined the theories of
computational monads and final coalgebras to obtain a
framework that generalizes process algebra to encompass processes with
 side-effecting steps. We have presented a sound and complete
equational calculus for the arising \emph{concurrent metalanguage} $\nuM$,
and we have obtained a syntactic corecursion scheme in which
corecursive functions are syntactically reducible to a basic loop
construct. Within this calculus, we have given generic definitions for
standard imperative constructs and a number of standard process
operators, most notably parallel composition.

Although the proof principles developed so far are already quite
powerful, as was shown in an example verification of a generic mutual
exclusion scheme following Dekker's algorithm, we intend to develop
more expressive verification logics for side-effecting processes,
detached from equational reasoning. Initial results of this kind have
already been used in the example verification, specifically an
encoding of generic Hoare triples and an associated proof principle
for safety properties. An interesting perspective in this direction is
to identify a variant of the assume/guarantee principle for
side-effecting processes (cf.~e.g.~\cite{RoeverBoerEtAl01}). A further
topic of investigation is to develop weak notions of process
equivalence in our framework, such as testing
equivalence~\cite{NicolaHennessy83}.

Finally, the decidability status of~\nuM remains open. Note that in
case of a positive answer, all equations between functions defined by
corecursion schemes, e.g.\ process algebra identities, become
decidable. While experience suggests that even very simple calculi
that combine loop constructs with monadic effects tend to be
undecidable, the corecursion axiom as a potential source of trouble
seems rather modest, and no evident encoding of an undecidable problem
appears to be directly applicable.

%\acks

%Acknowledgments, if needed.

% We recommend abbrvnat bibliography style.

\bibliographystyle{myabbrv}
\bibliography{monads}

\newpage
\appendix\allowdisplaybreaks
\section*{Appendix: Proof details.}
We justify the definition of the injections $\inj^n_i:A_i\to A_1+\cdots+A_n$ by the recursive equations:
\begin{flalign*}
&&\inj_1^1 p=p,&& \inj_1^{n+1} p=\inl p,&& \inj_{i+1}^{n+1} p=\inr\inj_i^n p.&&
\end{flalign*}
 We call a \emph{guarded corecursive equation} an equation of the form
\begin{align}\label{eq:wf_eq}\tag{$**$}
\unfold(f(x)) = \letTerm{z\leteq p}{\case{z}{\inj_1^{n} x_1\mapsto p_{1}}{\dotsc;\inj_{n}^{n} x_n\mapsto p_{n}}}
\end{align}
if $f$ does not occur in $p$ and none of the $p_i$ contains $f$ unless
${p_i\equiv\rcont f(x_i)}$ (which is, of course, well-typed only in
case $A_i=A$). The fact that the right-hand side is prefixed with the
binding $z\leteq p$ plays an important role for the expressiveness of
the scheme; a comparatively trivial point in this respect is that this
allows substituting the arguments $x_i$ in $\rcont f(x_i)$ by arbitrary
terms.
\begin{lem}[Corecursion]\label{lem:wf_eq}
  Given some appropriately typed programs $p$ and the~$p_i$ such that~\eqref{eq:wf_eq} is guarded, there is a unique function~$f$ satisfying~\eqref{eq:wf_eq} and this function is defined by an effectively computable metalanguage term.
\end{lem}
\begin{proof}
The idea is to start from a special case and successively extend generality.
\begin{enumerate}
\item[(i)] Suppose that $n=2$, $p_1\equiv\rcont f(x)$ and $p_2$ does not contain $f$. Let $A\to T_{\nu} B$ be the type profile of $f$ and let us define a function $F:A+T_{\nu} B\to T_{\nu} B$ by putting: $F(z) = \bigl(\letLoop{z\letret z}{H(z)}\bigr)$ where 
\begin{flalign*}
H(z)=\case{z}{\inl x\mapsto\bigl(&\letTerm{z\leteq p}{\case{z}{\\&%\hspace{11.2ex}
\inl x_1\mapsto\rcont(\inl x_1)}{\\&%\hspace{11.2ex}
												 \inr x_2\mapsto(\begin{array}[t]{rl}\casenext{p_2}{\cont x&\,\mapsto\rcont(\inr x)}{\\[1ex]
\done x&\,\mapsto\rdone x})}\bigr)}}{\end{array}\\
              \inr r\mapsto\bigl(&\begin{array}[t]{rl}\casenext{\unfold(r)}{\cont x&\,\mapsto\rcont (\inr x)}{\\[1ex]
																	\done x&\,\mapsto\rdone x}\bigr).\end{array}}
\end{flalign*}
Let us show that $F(\inr x)=x$. By~\textbf{(corec)}, 
\begin{flalign}\label{eq:F_H}
\begin{split}
\unfold(F(z)) = \casenext{H(z)}{\cont x\mapsto\rcont F(z)}{\done\mapsto\rdone x}.
\end{split}
\end{flalign}
Note that
\begin{flalign*}
H(\inr x) = \casenext{\unfold(x)}{\cont x\mapsto\rcont(\inr x)}{\done x\mapsto\rdone x}
\end{flalign*}
from which we conclude that
\begin{flalign*}
\unfold(F(\inr x)) = \casenext{\unfold(x)}{\cont x\mapsto\rcont F(\inr x)}{\done x\mapsto\rdone x}.
\end{flalign*}
The latter means that $(F\comp\inr)$ satisfies the same equation as the identity function. Hence, by~\textbf{(corec)} both these functions must be provably equal, i.e.\ for every $x$, $F(\inr x)=x$. Now,~\eqref{eq:F_H} can be simplified down to:
\begin{flalign*}
\unfold(F(z))=\case{z}{\inl x&\,\mapsto\bigl(
                                              \begin{array}[t]{rl}\casenext{p}{\cont x_1&\,\mapsto\rcont F(\inl(x_1))}{\\[1ex]
                                                              \done x_2&\,\mapsto p_2}\bigr)}{
                                              \end{array}
\\[-1.6ex]
           \inr r&\,\mapsto r}.
\end{flalign*}
By~\textbf{(corec)} $F$ is uniquely defined by this equation.
It can be verified by routine calculations that $f(x)=F(\inl x)$ is a solution of~\eqref{eq:wf_eq}. In order to prove uniqueness, let us assume that $g$ is some other solution of~\eqref{eq:wf_eq}. Let
\begin{flalign*}
G(z)=&\case{z}{\inl x\mapsto g(x)}{\inr r\mapsto\fold(r)}.
\end{flalign*}
Clearly, $g(x)=G(\inl x)$ and it can be shown that $G$ satisfies the equation defining $F$. Therefore $g(x)=G(\inl x)=F(\inl x)=f(x)$ and we are done. 
\item[(ii)] Let $n>1$, $p_1\equiv\rcont f(x_1)$ and for every $i>1$, $p_i$ does not contain $f$. We reduce this case to the previous one as follows. Observe that, by assumption,
\begin{equation*}\label{eq:wf_eq1}
\begin{split}
\unfold(f(x)) = \letTerm{z\leteq p}{\case{z}{\inl x_1\mapsto\rcont f(x_1)}{\inr z\mapsto q}}
\end{split}
\end{equation*}
where $q = (\case{z}{\inj_1^{n-1} x_2\mapsto p_2}{\dotsc;\inj_{n-1}^{n-1} x_{n}\mapsto p_{n}})$.  According to~(i), $f$ is uniquely definable and thus we are done.
\item[(iii)] Let for some index $k$, $p_i\equiv\rcont f(x)$ for $i\leq k$ and $p_i$ does not contain $f$ for $i>k$. If $k=1$ and $n>1$ then we arrive precisely at the situation captured by the previous clause and hence we are done. If $k=n=1$ then~\eqref{eq:wf_eq} takes the form:
\begin{displaymath}
\unfold(f(x)) = \letTerm{x\leteq p}{\rcont f(x)},
\end{displaymath}
which can be transformed to:
\begin{align*}
\unfold(f(x)) =\casenext{(\letTerm{x\leteq p}{\cont x})}{\cont x\mapsto\rcont f(x)}{\done x\mapsto\rdone x}
\end{align*}
and hence we are done by~\textbf{(corec)}.
Suppose that $k>1$, $n>1$ and proceed by induction over $k$. Let 
\begin{displaymath}
q = \case{z}{\inj_1^{n-2} x_3\mapsto p_3}{\dotsc;\inj_{n-2}^{n-2} x_{n-2}\mapsto p_{n-2}}.
\end{displaymath}
Then we have:
\begin{align*}
\unfold(f(x))=\;&\letTerm{z\leteq p}{\begin{array}[t]{rl}\case{z}{\inl x_1&\,\mapsto\rcont f(x_1)}{\\[1ex]\inr\inl x_2&\,\mapsto\rcont f(x_2);\\[1ex]\inr\inr z&\,\mapsto q}\end{array}}\\
=\;&\casenext{\begin{array}[t]{rl}\bigl(\letTerm{z\leteq p}{                                         
                                         \case{z}{
                                         \inl x_1&\,\mapsto\rcont x_1}{\\[1ex]
                                         \inr\inl x_2&\,\mapsto\rcont x_2;\\[1ex]
                                         \inr\inr z&\,\mapsto\rdone z}
                                         }\bigr)}{\end{array}\\
&\cont x\mapsto\rcont f(x)}{\done x\mapsto q}
\end{align*}
and thus we are done by induction hypothesis.
\item[(iv)] Finally, we reduce the general claim to the case captured by the previous clause as follows. First observe that if neither of the $p_i$ contains $f$ then the solution is given by the equation:
\begin{align*}
f(x) = \fold(\letTerm{z\leteq p}{\case{z}{\inj_1^{n} x_1\mapsto p_1}{\dotsc;\inj_{n}^{n} x_n\mapsto p_{n}}}).
\end{align*}
In the remaining case there should exist an index $k$ and a permutation $\sigma$ of numbers $1,\ldots,n$ such that for every $i\leq k$, $p_{\sigma(i)}\equiv\cont f(x)$ and for every $i>k$, $p_{\sigma(i)}$ does not contain $f$. By a slight abuse of notation we also use $\sigma$ as function $A_1+\ldots+A_n\to A_{\sigma(1)}+\ldots+A_{\sigma(n)}$ rearranging the components of coproducts in the obvious fashion. Then
{\samepage
\begin{flalign*}
\unfold(f(x)) =\;&\letTerm{z\leteq p}{\case{z}{\inj_1^{n} x_1\mapsto p_{1}}{\dotsc;\inj_{n}^{n} x_n\mapsto p_{n}}}\\
=\;&\letTerm{z\leteq(\letTerm{z\leteq p}{\ret\sigma(z)})}{\\&\case{z}{\inj_{\sigma(1)}^{n} x_{\sigma(1)}\mapsto p_{\sigma(1)}}{\dotsc;\inj_{\sigma(n)}^{n} x_{\sigma(n)}\mapsto p_{\sigma(n)}}}
\end{flalign*}
}
and thus we are done by~(iii).\qed
\end{enumerate}
\end{proof}
\subsection*{Proof of Theorem~\ref{lem:solve}. ~}
We will need the following slight generalisation of Lemma~\ref{lem:wf_eq}.
\begin{lemma}\label{lem:solve_gen}
  Let $f$ be a fresh functional symbol, i.e.\ $f\notin\Sigma$. Given
  appropriately typed programs $p,p_1,\dots,p_n,q_1,\dots,q_n$ such
  that for every $i$, $p_i$ either does not contain $f$ or is of the
  form $\rcont f(q_i)$, there is a unique function~$f$
  satisfying~\eqref{eq:wf_eq}.
\end{lemma}
\begin{proof}
W.l.o.g.\ $q_i\equiv x$ whenever $p_i$ does not contain $f$. We can rewrite~\eqref{eq:wf_eq} to
\begin{align*}
\unfold(f(x)) = \letTerm{z\leteq q}{\case{z}{\inl_1^n x_1\mapsto r_1}{\dotsc;\inl_1^n x_n\mapsto r_n}}
\end{align*}
where $q=(\letTerm{z\leteq p}{\case{z}{\inj_1^n x_1\mapsto\ret\inj_1^n q_1}{\dotsc;\inj_1^n x_n\mapsto\ret\inj_1^n q_n}})$, $r_i=\rcont(f(x))$ if $p_i\equiv\rcont(f(q_i))$ and $r_i=p_i$ otherwise. Now we are done by Lemma~\ref{lem:wf_eq}.
\qed\end{proof}
W.l.o.g.\ $n_1=\ldots=n_k$: otherwise we replace every $p_i$ by
\begin{align*}
\letTerm{z\leteq p_i}{\case{z}{\inj_1^{n_i} x_1\mapsto\inj_1^{n} x_1}{\dotsc;\inj_{n_i}^{n_i} x_{n_i}\mapsto\inj_{n_i}^{n} x_n}}
\end{align*}
where $n=\max_i n_i$ and complete the case in~\eqref{eq:corec} arbitrary to match the typing. Observe that if the result types of $f_i$ do not coincide,~\eqref{eq:corec} falls in two mutually independent parts: once the result types of $f_i$, $f_k$ are distinct, $f_k$ can not appear in the definition of $f_i$ and vice versa. Therefore, in the remainder we assume w.l.o.g.\ that $f_i$ has type $A_i\to A$. 
%We proceed by induction over $n$. If $n=1$ then we are done by Lemma~\ref{lem:wf_eq}. Let us consider the case $n>1$. For every $i<n$ let $g_i:A_i+A_n\to T_{\nu}(B_i+B_n)$ be a fresh functional symbol. By induction, a system of equations
Consider the following corecursive definition:
\begin{align*}
\unfold(F(y)) =  \letTerm{v\leteq q}{&\case{v}{\\\inj_1^k z&\mapsto(\case{z}{\inj_1^n x_1\mapsto q_{1}^1}{\dotsc;\inj_{n}^{n} x_n\mapsto q_{n}^1})}{\\
&\vdots\\
                                                     \inj_k^k z&\mapsto(\case{z}{\inj_1^n x_1\mapsto q_{1}^k}{\dotsc;\inj_{n}^{n} x_n\mapsto q_{n}^k})}}
\end{align*}
where
\begin{flalign*}
q=\case{y}{\inj^k_1 x_1&\mapsto(\letTerm{z\leteq p_i}{\ret\inj^k_1 z})}{\dotsc;
\inj_k^k x_k\mapsto(\letTerm{z\leteq p_k}{\ret\inj^k_k z})}
\end{flalign*}
and the $q_j^i$ are defined as follows:
$q_j^i=\rcont F(\inj_m^k x_j)$ if $p_j^i\equiv\rcont f_m(x_j)$ and $q_j^i=p_j^i$ otherwise. By Lemma~\ref{lem:solve_gen}, it uniquely defines $F$. It is easy to calculate that by taking $f_i(x)=F(\inj_i^k(x))$ we obtain a solution of~\eqref{eq:corec}. Let us show it is also unique. Suppose that $g_i$ is another solution. Then $G$ defined by the equation 
\begin{align*}
G(y) = \case{y}{\inj_1^k x\mapsto g_1(x)}{\dotsc;\inj_k^k x\mapsto g_k(x)}
\end{align*}
is easily seen to satisfy the same corecursive scheme as $F$ and thus $F=G$. Therefore, for every $i$, $g_i(x)=G(\inj_i^k x)=F(\inj_i^k x)=f_i(x)$ and we are done.
\qed

\subsection*{Proof of Theorem~\ref{thm:rec_compl}. ~}
\textit{Soundness.} We only establish soundness of the rule~\textbf{(corec)} since the remainder is more or less standard.
Let $g=\sem{\Gamma,x:A\rhd q:T(A+B)}$, $h=\sem{\Gamma,x:A\rhd p:T_{\nu} B}$ and $f=T(\mathit{dist})\comp\tau\brks{\pi_1,g}$. Then the top of~\textbf{(corec)} can be rewritten to
\begin{equation}\label{eq:corec_proof1}
h=R\pi_2\comp\ana{f}.
\end{equation}
Let us show that the bottom of~\textbf{(corec)} can be rewritten to
%%
%\begin{equation}\label{eq:corec_proof2}
%\fin{\sem{B}}\comp h = T\bigl((h\comp(\pi_1\pi_1\times\id)+\pi_2)\comp \mathit{dist}\comp\brks{\id,\pi_2}\bigr)\comp\tau\brks{\id,f}.
%\end{equation}
%
\begin{equation}\label{eq:corec_proof2}
\fin{\sem{B}}\comp h = T(h+\pi_2)\comp f.
\end{equation}
Indeed, we have:
\begin{flalign*}
\fin{\sem{B}}\comp h
=\;&\sem{\unfold(p)}\\
=\;&\sem{\letTerm{z\leteq q}{\ \case{z}{\inl x\mapsto\ret\inl p}{\inr y\mapsto\ret\inr y}}}\\
=\;&T\bigl((h\comp(\pi_1\pi_1\times\id)+\pi_2)\comp \mathit{dist}\comp\brks{\id,\pi_2}\bigr)\comp\tau\brks{\id,g}\\
=\;& T\bigl((h+\pi_2)\comp(\pi_1\pi_1\times\id+\pi_1\pi_1\times\id)\comp \mathit{dist}\comp\brks{\id,\pi_2}\bigr)\comp\tau\brks{\id,g}\\
=\;& T\bigl((h+\pi_2)\comp \mathit{dist}\comp (\pi_1\pi_1\times\id)\comp\brks{\id,\pi_2}\bigr)\comp\tau\brks{\id,g}\\
=\;& T\bigl((h+\pi_2)\comp \mathit{dist}\comp (\pi_1\times\id)\bigr)\comp\tau\brks{\id,g}\\
=\;& T(h+\pi_2)\comp T(\mathit{dist})\comp\tau\brks{\pi_1,g}\\
=\;& T(h+\pi_2)\comp f
\end{flalign*}

In order to complete the proof, we are left to establish equivalence of \eqref{eq:corec_proof1} and \eqref{eq:corec_proof2}. The proof of the implication~\eqref{eq:corec_proof1} $\impl$ \eqref{eq:corec_proof2} is as follows:
\begin{flalign*}
\fin{\sem{B}}\comp h
=\;&\fin{\sem{B}}\comp R\pi_2\comp\ana{f}\\
=\;& T(R\pi_2+\pi_2)\comp\fin{\sem{\Gamma}\times\sem{B}}\comp\ana{f}\\
=\;& T(R\pi_2+\pi_2)\comp T(\ana{f}+\id)\comp T(\mathit{dist})\comp\tau\brks{\pi_1,g}\\
=\;& T(R\pi_2\comp\ana{f}+\pi_2)\comp T(\mathit{dist})\comp\tau\brks{\pi_1,g}\\
=\;& T(h+\pi_2)\comp T(\mathit{dist})\comp\tau\brks{\pi_1,g}\\
=\;& T(h+\pi_2)\comp f.
\end{flalign*}
In order to prove \eqref{eq:corec_proof2}, $\impl$ \eqref{eq:corec_proof1} let us assume~\eqref{eq:corec_proof2}. Observe that we can equivalently present the latter as $\fin{\sem{B}}\comp h = T(h+\id)\comp w$ where $w=T(\id+\pi_2)\comp f$. I.e.\ $h$ satisfies the equation, which characterises $\ana{f}$ and thus $h=\ana{w}$. We are left to show that $R\pi_2\comp\ana{f}$ also satisfies this equation. Indeed: 
\begin{flalign*}
\fin{\sem{B}}\comp &R\pi_2\comp\ana{f}\\
=\;& T(R\pi_2+\pi_2)\comp\fin{\sem{\Gamma}\times\sem{B}}\comp\ana{f}\\
=\;& T(R\pi_2+\pi_2)\comp T(\ana{f}+\id)\comp T(\mathit{dist})\comp\tau\brks{\pi_1,g}\\
=\;&T(R\pi_2\comp\ana{f}+\pi_2)\comp T(\mathit{dist})\comp\tau\brks{\pi_1,g}\\
=\;&T(R\pi_2\comp\ana{f}+\pi_2)\comp f\\
=\;&T(R\pi_2\comp\ana{f}+\id)\comp T(\id+\pi_2)\comp f\\
=\;&T(R\pi_2\comp\ana{f}+\id)\comp w.
\end{flalign*}
We have thus $h=\ana{w}=R\pi_2\comp\ana{f}$ and the proof is completed.

\textit{Completeness.} By term model construction. \qed

\subsection*{Proof of Lemma~\ref{lem:hoare}. ~}
Suppose that $\{\phi\}p\{\psi\}$. Then we have:
\begin{flalign*}
\letTerm{x\leteq&\,\phi;y\leteq p;z\leteq\psi}{\ret\brks{x,y,z,x\impl z}}\\
=\,&\letTerm{x\leteq\phi;y\leteq\filter(\phi,p,\psi);z\leteq\psi}{\ret\brks{x,y,z,x\impl z}}\\
=\,&\letTerm{x\leteq\phi;y\leteq(\letTerm{x'\leteq\phi;y'\leteq p;z'\leteq\psi}{\\&\ifTerm{(x'\impl z')}{\ret y'}{\nil}});z\leteq\psi}{\ret\brks{x,y,z,x\impl z}}\\
=\,&\letTerm{x\leteq\phi;x'\leteq\phi;y'\leteq p;z'\leteq\psi}{\\&\ifTerm{(x'\impl z')}{\letTerm{y\leteq\ret y';z\leteq\psi}{\ret\brks{x,y,z,x\impl z}}\\&\hspace{11.6ex}}
{\letTerm{y\leteq\nil;z\leteq\psi}{\ret\brks{x,y,z,x\impl z}}}}\\
=\,&\letTerm{x\leteq\phi;y\leteq p;z\leteq\psi}{\ifTerm{(x\impl z)}{\ret\brks{x,y,z,x\impl z}}{\nil}}\\
=\,&\letTerm{x\leteq\phi;y\leteq p;z\leteq\psi}{\ifTerm{(x\impl z)}{\ret\brks{x,y,z,\True}}{\nil}}\\
=\,&\letTerm{x\leteq\phi;y\leteq p;z\leteq\psi}{\ret\brks{x,y,z,\True}}
\end{flalign*}
On the other hand, provided the equation
\begin{align}\label{eq:hoare}
\begin{split}
&\letTerm{x\leteq\phi;y\leteq p;z\leteq\psi}{\ret\brks{x,y,z,x\impl z}}=\\
&\letTerm{x\leteq\phi;y\leteq p;z\leteq\psi}{\ret\brks{x,y,z,\True}},
\end{split}
\end{align}
we have:
\begin{flalign*}
\filter(&\phi,p,\psi)\\
=\,&\letTerm{x\leteq\phi;y\leteq p;z\leteq\psi}{\ifTerm{(x\impl z)}{\ret y}{\nil}}\\
=\,&\letTerm{\brks{x,y,z,v}\leteq(\letTerm{x\leteq\phi;y\leteq p;z\leteq\psi}{\ret(x,y,z,x\impl z)})}{\\&\ifTerm{v}{\ret y}{\nil}}\\
=\,&\letTerm{\brks{x,y,z,v}\leteq(\letTerm{x\leteq\phi;y\leteq p;z\leteq\psi}{\ret(x,y,z,\True)})}{\\&\ifTerm{v}{\ret y}{\nil}}\\
=\,&\letTerm{x\leteq\phi;y\leteq p;z\leteq\psi}{\ifTerm{\True}{\ret y}{\nil}}\\
=\,&p.
\end{flalign*}
By definition, this means validity of the Hoare triple $\{\phi\}p\{\psi\}$.\qed
\end{document}

%%% Local Variables:
%%% TeX-PDF-mode:t
%%% End:

% LocalWords:  expl